\documentclass[10pt]{./llncs2e/llncs}
\usepackage{amsmath, amssymb}
\usepackage{lipsum}
\usepackage{color}
\usepackage{hyperref}
\usepackage{graphicx}
\usepackage{mathtools}

\usepackage{caption}
\usepackage{subcaption}

\usepackage{color}
\usepackage{bbm}
\usepackage{soul}

\usepackage{enumitem}
\usepackage{multirow}

\usepackage[algoruled,vlined]{algorithm2e}
\usepackage{algpseudocode}
\usepackage{nicefrac}

\usepackage{floatrow}
\DeclareFloatFont{tiny}{\tiny}% "scriptsize" is defined by floatrow, "tiny" not
\newcommand{\paran}[1]{\left(#1\right)}
\newcommand{\set}[1]{\left\{#1\right\}}

\newcommand{\E}[1]{\mathbb{E}\left[#1\right]}
\newcommand{\pr}[1]{\text{Pr}\left[#1\right]}
\newcommand{\cost}[1]{\text{cost}\left( #1 \right)}

\DeclareMathOperator*{\argmin}{argmin}

\def\S{\mathcal{S}}

\def\D{\mathcal{D}}
\def\R{\mathcal{R}}
\def\O{\mathcal{O}}

\def\ada{\texttt{Adaptive}}

\def\tpi{\tilde \pi}

\def\S{\mathcal{S}}
\def\D{\mathcal{D}}

\def\O{\mathcal{O}}

\def\R{\mathcal{R}}

\def\A{\mathcal{A}}

\def\tpi{\tilde \pi}
\def\ps{p^*}

\begin{document}
%\title{\ahmad{Toward efficient monitoring of massive dynamic data}}

\title{
%\ahmad{Schedule Optimization Problem} \\
Optimizing Static and Adaptive Probing Schedules for Rapid Event Detection}

\author{Ahmad Mahmoody \and Evgenios M. Kornaropoulos \and Eli Upfal}
\institute{Department of Computer Science, Brown University \\\email{\{ahmad, evgenios, eli\}@cs.brown.edu} }

\maketitle

\begin{abstract}
We formulate and study a fundamental search and detection problem, \emph{Schedule Optimization}, motivated by a variety of real-world applications, ranging from monitoring content changes on the web, social networks, and user activities to detecting failure on large systems with many individual machines.

We consider a large system consists of many nodes, where each node has its own rate of generating new events, or items. A monitoring application can probe a small number of nodes at each step, and our goal is to compute a probing schedule that minimizes the expected number of undiscovered items at the system, or equivalently, minimizes  the expected time to discover a new item in the system.

We study the Schedule Optimization problem both for deterministic and randomized memoryless algorithms. We provide lower bounds on the cost of an optimal schedule and construct close to optimal schedules with rigorous mathematical guarantees. Finally, we present an adaptive algorithm that starts with no prior information on the system and converges to the optimal memoryless algorithms by adapting to observed data.
\end{abstract}

\section{Introduction}

We introduce and study a fundamental stochastic search and detection problem, \emph{Schedule Optimization}, that captures a variety
of practical applications, ranging from monitoring content changes on the web, social networks, and user activities to detecting failure on large systems with many individual machines. 

Our optimization problem consists  of a large set of units, or \emph{nodes}, that generate events, or \emph{items}, according to a random process with known or unknown parameters. A detection algorithm can discover new items in the system by probing a small number of nodes in each step. This setting defines a discrete, infinite time process, and the goal of the stochastic optimization problem is to construct a probing schedule that minimizes the long term expected number of undiscovered items in the system, or equivalently, minimizes  the expected time to discover a new item in the system.

We outline several important applications of this schedule optimization problem:

%In today's world, data is ubiquitous in every aspects of our lives. Data are generated in high volume and velocity\footnote{... and many other V's\cite{}.}. They usually get generated constantly by many different resources\cite{}. For instance, online users browse the internet and search for different queries, and all these activities are generated by millions of users constantly and result in a huge amount of data that constantly get generated.
%
%However, these data might be considered as \emph{raw} data that need to be processed further for analytical tasks to acquire more information. This happens in a wide range of applications, and in following we point out just to some of these applications.
%
\paragraph{News and Feed Aggregators.} To provide up to date summary of the news, news aggregator  sites need to constantly browse the Web, and often also the blogosphere and social networks, for new items. Scanning a site for new items requires significant communication and computation resources, thus the news aggregator can scan only a few sites simultaneously. The frequency of visiting a site has to depend on the likelihood of finding new items in that site. \cite{onlineRef-horincar2014online,fast-sia2007efficient,adaptive-bright2006adaptive}

\paragraph{Algorithmic Trading on Data.}
An emerging trend in algorithmic stock trading is the use of automatic search through the Web, the blogosphere, and social networks for relevant information that can be used in fast trading, before it appears in the more popular news sites~\cite{Delaney2009,ALPHA2014,AlphaFlash,mitra2011handbook,latar2015robot,wallstreet2015,McKinney2011}. The critical issue in this application is the speed of discovering new events, but again there is a resource limit on the number of sites that the search algorithm can scan simultaneously. 

%\paragraph{Search Data.}
%Online users constantly search and browse the Internet. These \emph{raw} history data are gathered (probably by third channel parties \cite{}) to analyze and predict the behavior and needs of the users, in order to improve their services, whether by showing more relevant ads or search results. 
%
%\paragraph{Algorithmic Stock Trading.}
%An emerging trend in algorithmic stock trading is the use of automatic search through the Web, the blogosphere, and social networks for relevant information that can be used in fast trading, before it appears in the more popular news sites~\cite{Delaney2009}. 

\paragraph{Detecting Anomaly and Machine Malfunction.}
In large server farm or any other large collection of semi-autonomous machines a central controller needs to identify and contain anomalies and malefactions as soon as possible, before they spread in the system. To minimize interference with the system's operation the controller must probe only a small number of machines in each step. 

\subsection{Our Contribution}

We consider an infinite, discrete time process in which $n$ nodes generate new items according to a stochastic process which is governed by a generating  vector $\pi$ (see Section~\ref{sec:model} for details). An algorithm can probe up to $c$ nodes per step to discover all new items in these nodes. The goal is to minimize the \emph{cost} of the algorithm (or the probing schedule), which we define as the long term (steady state) expected number of undiscovered items in the system.

We first show that the obvious approach of probing at each step the nodes with maximum expected number of undiscovered items at that step is not optimal. In fact, the cost of such a schedule can be arbitrary far from the optimal. 

Our first result toward the study of efficient schedules is a lower bound on the cost of any deterministic or random schedule as a function of the generating vector $\pi$.

Next we assume that the generating vector $\pi$ is known and study explicit constructions of deterministic and random schedules. We construct a deterministic schedule whose cost is within a factor of  $\paran{3+(c-1)/c}$ of the optimal cost, and a very simple, memoryless random schedule with cost that is within a factor of $\paran{2+(c-1)/c}$ from optimal, where $c$ is the maximum number of probes at each step.

Finally, we address the more realistic scenario in which the generating vector, $\pi$, is not known to the algorithm and may change in time. We construct an adaptive scheduling algorithm that learns from probing the nodes and converges to the optimal memoryless random schedule.

\section{Related Work}\label{sec:related_work}
%\ahmad{In this section we present some other work/problems that have some similar flavor to our problem, \emph{Schedule Optimization}. }

The \emph{News and Feed Aggregation} problem is a very well-studied topic, in which the general goal is to obtain the updates of news websites (e.g. by RSS feeds). Among many introduced objectives~\cite{survey-oita2011deriving,onlineRef-horincar2014online,adaptive-bright2006adaptive} in studying this problem, the most similar one to our cost function is the \emph{delay} function presented by~\cite{fast-sia2007efficient}.  In~\cite{fast-sia2007efficient} it is assumed that the rates of the news publication does not change, where in our setting these rates may change and our algorithm ({\ada}) can adapt itself to the new setting. Also, we assume at any given time the number of probes is fixed (or bounded) regarding the limited computational power for simultaneous probes, but \cite{fast-sia2007efficient} uses a relaxed  assumption by fixing the number of probes over a \emph{time window} of a fixed length which may result in high number of probes at a single time step. Finally, \cite{fast-sia2007efficient} introduces a deterministic algorithm in which the number of probes to each feed is obtained by applying the Lagrange multipliers method (very similar result to Theorem~\ref{thm:randomized_schedule}), but they loose the guarantee on optimality of their solution, by rounding the estimated number of probes to integers. In contrast, our solution provides  theoretical guarantee on optimality of our output schedule.

\emph{Web-crawling} is another related topic, where a web-crawler aims to obtain the most recent snapshots of the web. However, it differs from our model substantially: in web-crawling algorithm data get \emph{updated}, so missing some intermediate snapshot would not affect the quality of the algorithm, where in our model data are generated and they all need to be processed~\cite{dasgupta2007discoverability,wolf2002optimal}.

There has been an extensive work on  \emph{Outbreak Detection} (motivated in part by the ``Battle of Water Sensors Network" challenge~\cite{BWSN2008}) using statistic or mobile sensor in physical domains, and regarding a variety of objectives~\cite{Leskovec2007,Krause2008,Hart2010}. Our model deviates from the Outbreak Detection problem as it is geared to detection in virtual networks such as the Web or social networks embedded in the Internet, where a monitor can reach (almost) any node at about the same cost. 

Another related problem is the \emph{Emerging Topic Detection} problem, where the goal is to identify emergent topics in a social network, assuming full access to the stream of all postings. Besides having different objectives, our model differs mainly in this accessibility assumption: the social network providers have an immediate access to all tweets or postings as they are submitted to their servers, whereas in our model we consider an outside observer who  needs an efficient mechanism to monitor changes, without having such full access privilege~\cite{Cataldi2010,Mathioudakis2010}.

In the next section, we formally define our model and the Schedule Optimization problem.

\section{Model and Problem Definition}\label{sec:model}
We study an infinite, discrete time process in which a set of $n$ \emph{nodes}, indexed by $1,\ldots,n$, generate new \emph{item}s according to a random generating process. The generating process at a given time step is characterized by a \emph{generating vector} $\pi=(\pi_1,\ldots,\pi_n)$, where 
$\pi_i$ is the expected number of new items generated at node $i$ at that step (by either a Bernoulli or  a Poisson process). The generation processes in different nodes are independent.

%Our results extend to a more general generating process in which at each step node $i$ generates a number of new items according to Poisson distribution with expectation $\pi_i$.
 We focus first on a \emph{static generating process} in which the generating vector does not change in time.  We then extend our results to adapt to  generating vectors that change in time.

%The system we are studying consisted of $n$ \emph{node} which generate \emph{item} according to a Our goal is to discover new \emph{items}
%In this section, we introduce our model and formally define our problem. Throughout this work, we call each distributed resource (that generates the data) by a \emph{node}, and assume that there are $n$ nodes indexed by $1,\ldots,n$. We also call the units of data generated at nodes as \emph{item}. As mentioned above, our model occurs in an infinite time process, where at discrete time steps new items get generated at nodes, and we can probe up to $c$ nodes at each step. A \emph{generating vector} $\pi=(\pi_1,\ldots,\pi_n)$ is a vector where $\pi_i$ is the probability that node $i$ generates a new item at each time step. When a node is probed, we \emph{catch} all the items generated so far in that node. When an item is not caught yet, we may say it is \emph{undiscovered} or \emph{undetected}. Now, let's define the schedules:

Our goal is to detect new events as fast as possible by probing in each step a small number of nodes.
In particular, we consider probing schedules that can probe up to $c$ nodes per step.

\begin{definition}[Schedule] A \emph{$c$-schedule} is a function $\S:\mathbb{N} \rightarrow \{1,\dots,n\}^c$ specifying a set of $c$ nodes to be probed at any time $t\in \mathbb{N}$.
A deterministic function $\S$ defines a \emph{deterministic} schedule, otherwise the schedule is \emph{random}.
%	Suppose $\mathcal{X}$ is the set of all possible \ahmad{outcomes of a source randomness}, and let $P(n,c)$ be the set of all subsets of $[n]$ with at most $c$ elements. A \emph{$c$-schedule} is a function $\S: \mathcal{X}\times \mathbb{N} \rightarrow P(n,c)$, where $\S(x,t)$ assigns a subset of at most $c$ nodes to be probed at time $t$ (with access to the random coin $x$). A schedule $\S(x,t)$ is called (i) \emph{deterministic} if is independent of $x$, and (ii) \emph{memoryless} if it is independent of $t$ and there is a distribution $p=(p_1,\ldots,p_n)$ such that $\S(x,t)$ is a random draw from $p$; in this case, for notational convenience we denote $\S = p$. A schedule is \emph{cyclic} if it is deterministic and gets periodic after a time $t_0$, and it is $\ell$-cyclic if the length of its period is $\ell$.
\end{definition}

\begin{definition}[Memoryless Schedule]
A random schedule is \emph{memoryless} if it is defined by a vector $p=(p_1,\dots,p_n)$ such that at any step the schedule probes a set $C$ of $c$ items with
probability $\prod_{j\in C} p_i$ independent of any other event. In that case we use the notation $\S=p$.
\end{definition}

\begin{definition}[Cyclic Schedule]
A schedule, $\S$, is \emph{$\ell$-cyclic} if there is a finite time $t_0$ such that from time $t_0$ on, the schedule repeats itself every period of $\ell$ steps. 
A schedule is cyclic if it is $\ell$-cyclic for some positive integer $\ell$.
\end{definition}

%The monitoring algorithm can probe $c$ nodes at each step and our goal is to construct a probing schedule that discovers all new items in the network as fast as possible. We consider both deterministic and randomized (memoryless) schedules.

%\begin{definition}
%A \emph{deterministic $c$-schedule} is a function  
%$\D: \mathbb{N}\rightarrow \set{1,\ldots,n}^c$,
%assigning $c$ nodes to be probed at each time step $t\in\mathbb{N}$. \ahmad{The schedule is \emph{$\ell$-cyclic} if it repeats itself every $\ell$ steps, and is \emph{cyclic} if it is $\ell$-cycle for some positive integer $\ell$.}
%\end{definition}
%
%\begin{definition}
%A \emph{randomized, memoryless $c$-schedule}, or simply a randomized $c$-schedule, is a probability distribution $\mathbf{p} = (p_1,\ldots,p_n)$ over $\set{1,\ldots,n}$,
%such that at each step the schedule probes $c$ nodes chosen independently according to the distribution $\mathbf{p}$.
%\end{definition}

%Note that storing and/or computing  an optimal, or close to optimal, deterministic schedule can be very expensive. In fact, as in Theorem~\ref{thm:cyclic}, every optimal deterministic schedule is cyclic. However, the the period of the algorithm can be very long, and thus, the algorithm can be very expensive to store or execute. Instead, we show that simple to implement randomized, memoryless schedules provide an efficient alternative to the optimal deterministic schedule.

The quality of a probing schedule is measured by the speed in which it discovers new items in the system.
When a schedule probes a node $i$ at a time $t$, all items that were generated at that node by time $t-1$ are discovered (thus, each item is not discovered in at least one step). We define the \emph{cost} of a probing schedule as the long term expected number of undiscovered items in the system.

\begin{definition}[Cost]
The \emph{cost} of schedule $\S$ in a system of $n$ nodes with generating vector $\pi$ is
 $$\cost{\S,\pi} = \lim_{t\rightarrow \infty}  \frac{1}{t}\sum_{t'=1}^t \E{Q^{\S}(t')} 
 = \lim_{t\rightarrow \infty}  \frac{1}{t}\sum_{t'=1}^t \sum_{i=1}^n \E{Q_i^{\S}(t')},
 $$
where $Q_i^{\S}(t')$ is the number of undiscovered items at node $i$ and at time $t'$, and $Q^\S(t') = \sum_{i=1}^n Q_i^\S(t')$. 
The expectation is taken over the distribution of the generating system and the probing schedule. 
\end{definition}
While the cost can be unbounded for some schedules, the cost of the optimal schedule is always bounded. To see that, consider a round-robin schedule, $\S$, that probes each node every $n$ steps. Clearly no item is undiscovered in this schedule for more than $n$ steps, and the expected number of items generated in an interval of $n$ steps is $n\sum_{i=1}^n \pi_i$. Thus, $Q^\S(t) \leq n\sum_{i=1}^n \pi_i$, which implies $\cost{\S,\pi} \leq n\sum_{i=1}^n \pi_i$. Therefore, without loss of generality we can restrict our discussion to  bounded cost schedules. Also, note  that when the sequence $\set{\E{Q^\S(t)}}_{t\in\mathbb{N}}$ converges we have  $\cost{\S,\pi} = \lim\limits_{t\rightarrow\infty} \E{Q^\S(t)}$ (Cesaro Means~\cite{hardy1991divergent}).
%\end{lemma}
%\begin{proof}
%	For a convergent sequence $\set{a_t}_{t\in\mathbb{N}}$ we have $\lim\limits_{t\rightarrow\infty}\frac{1}{t}\sum_{t'=1}^t a_{t'} = \lim\limits_{t\rightarrow\infty}a_t$. Now by letting $a_t = \E{Q^{\S}(t)}$, the proof is complete. \qed
%\end{proof}

One can equivalently define the cost of a schedule in terms of the expected time that an item is in the system until it is discovered.

%$\sum_{i=1}^n \pi_i \omega^\S_i$, where $\omega^S_i$ is the average waiting time of an item generated at node $i$ until $i$ is probed.
\begin{lemma}~\label{lem:equivalent_cost}
Let $\omega^S_i$ be the expected waiting time of an item generated at node $i$ until node $i$ is probed by schedule $\S$. Then
$$\cost{\S,\pi} = \sum_{i=1}^n \pi_i \omega^\S_i.$$
\end{lemma}
\begin{proof}
Following the definition of the cost function we have
\begin{eqnarray*}\label{eq:cost}
	\cost{\S,\pi} %=  \lim_{t\rightarrow \infty} \E{R^{\S}(t)}
	=\lim_{t\rightarrow\infty} \frac{1}{t} \sum_{t'=1}^{t}\sum_{i=1}^n \E{Q_i^\S(t')} 
	=\sum_{i=1}^n \left[ \lim_{t\rightarrow\infty} \frac{\sum_{t'=1}^{t}\E{Q_i^\S(t')}}{t}\right] \nonumber 
	= \sum_{i=1}^n \pi_i \omega^\S_i,
\end{eqnarray*}
%Suppose $Q^\S_i(t)$ is the number of undiscovered items at node $i$ and at time $t$. Obviousely $Q^{\S}(t) = \sum_{i=1}^n Q^\S_i(t)$ and we have
%\begin{eqnarray*}\label{eq:cost}
%	\cost{\S} %=  \lim_{t\rightarrow \infty} \E{R^{\S}(t)}
%	&=&
%	\lim_{t\rightarrow\infty} \frac{1}{t} \sum_{t'=1}^{t} \E{Q^\S(t')} \nonumber 
%	=\lim_{t\rightarrow\infty} \frac{1}{t} \sum_{t'=1}^{t}\sum_{i=1}^n \E{Q_i^\S(t')} \\
%	&= &\sum_{i=1}^n \left[ \lim_{t\rightarrow\infty} \frac{\sum_{t'=1}^{t}\E{Q_i^\S(t')}}{t}\right] \nonumber 
%	= \sum_{i=1}^n \pi_i \omega^\S_i,
%\end{eqnarray*}\qed
where the last eqaulity is obtained by applying Little's Law~\cite{Leon-Garcia2008}.  \qed
\end{proof}

\begin{corollary}
A schedule that minimizes the expected number of undiscovered items in the system simultaneously minimizes the expected time that an item is undiscovered.
\end{corollary}

\begin{corollary}\label{cor:sum_lower}
	For any schedule $\S$, $\cost{\S, \pi} \geq \sum_{i=1}^n \pi_i$.
\end{corollary}
\begin{proof}
	As mentioned above, when we probe a node $i$ at time $t$ we discover only the items that have been generated by time $t-1$. Therefore, $\omega_i^\S \geq 1$, and by Lemma~\ref{lem:equivalent_cost} the proof is complete. \qed
\end{proof}

%Thus, we focus on the follows:
Now, our main problem is defined as the following:
\begin{definition}[Schedule Optimization] Given a generating vector $\pi$ and a positive integer $c$, find a  $c$-schedule with minimum cost. 
\end{definition}
%As we later show (Theorem~\ref{thm:cyclic}) there always exists a cyclic optimal schedule. However, storing and/or computing  an optimal, or close to optimal, deterministic schedule can be very expensive.
%

When the generating vector is not known a priori to the algorithm the goal is to design a schedule that \emph{converges} to an optimal one. For that we need the following definition:
\begin{definition}[Convergence]
	We say schedule $\S$ \emph{converges} to schedule $\S'$, if for any generating vector $\pi$,
	$\lim\limits_{t\rightarrow\infty} \left|\E{Q^{\S}(t)} - \E{Q^{\S'}(t)}\right| = 0$.
%	if there is a function $m:(0,1)^2 \rightarrow \mathbb{N}$ such that if $t \geq m(\epsilon,\delta)$ with probability at least $1-\delta$ we have $|Q^{\S}(t) - Q^{\S'}(t)| < \epsilon \cdot Q^{\S'}(t)$.
\end{definition}

\section{Results}\label{sec:isolated_nodes}
We start this section by, first, showing that the obvious approach of maximizing the expected number of detections at each step is far from optimal. We then prove a lower bound on the cost of any schedule, and provide  deterministic and memoryless $c$-schedules that are within a factor of $(3+(c-1)/c)$ and $(2+(c-1)/c)$, respectively, from the optimal. Finally, we introduce an algorithm, {\ada} , which outputs a schedule $\A$ that converges to the optimal \emph{memoryless} 1-schedule when the generating vector $\pi$ is not known in advance. We also show that {\ada} can be used to obtain a $c$-schedule $\A^c$  whose cost is within $(2+(c-1)/c)$ factor of any optimal $c$-schedule.

Throughout this section, by $\tau^\S_i(t)$ we mean the number of steps from the last time that node $i$ was probed until time $t$, while executing schedule $\S$; if $i$ has not been probed so far, we let $\tau_i^\S(t) = t$. Using the definition, it is easy to see that
\begin{equation}\label{eq:tau_q}
\E{Q^\S_i(t)} = \pi_i \E{\tau_i^\S(t)},
\end{equation}
when the expectations are over the randomness of \emph{both} $\S$ and $\pi$. Therefore, if the expectation is over \emph{only} the randomness of $\pi$ we have 
\begin{equation}\label{eq:tau_q_pi_rand}
\E{Q^\S_i(t)} = \pi_i \tau_i^\S(t).
\end{equation}

%%%%%%%%%%%%%%%%%%%%%%%%%%%%%%%%%%%%%%%%%%%%%
% Immediate Gain
%%%%%%%%%%%%%%%%%%%%%%%%%%%%%%%%%%%%%%%%%%%%%

\subsection{On Maximizing Immediate Gain} 
Let $\S$ be a 1-schedule that at each step, probes the node with the maximum expected number of undetected items. By \eqref{eq:tau_q_pi_rand}, the expected number of undetected items at node $i$ and at time $t$ is $\pi_i\tau_i^\S(t)$, and thus, $\S(t)=\arg\max_i \pi_i \tau^\S_i(t)$.

Now, suppose $\pi_i= 2^{-i}$, for $1\leq i \leq n$. Since the probability that node 1 has an undetected item in each step is at least $1/2$, node $i$ is probed no more than once in each $2^{i-1}$ steps. Thus, the expected number of time steps that an item at node $i$ will stay undetected is at least
$\frac{1}{2^{i-1}}(1+\ldots+2^{i-1}) = \frac{2^{i-1}+1}{2} > 2^{i-2}$.
Using Lemma~\ref{lem:equivalent_cost}, the cost of this schedule is at least 
$\sum_{i=1}^n \pi_i \omega_i > \sum_{i=1}^n  2^{-i}2^{i-2} =\Omega(n)$. Now, consider an alternative schedule that probes node $i$ in each step with probability $2^{-i/2}/Z$, where $Z=\sum_{j=1}^n 2^{-j/2}$. The expected number of steps between two probes of $i$ is $Z/2^{-i/2}$, and the cost of this schedule is 
$$\sum_{i=1}^n 2^{-i} \paran{\frac{2^{-i/2}}{\sum_{j=1}^n 2^{-j/2}}}^{-1}= \paran{\sum_{j=1}^n 2^{-j/2}}^2 = O(1).$$
Thus, optimizing immediate gain is not optimal in this problem.

%%%%%%%%%%%%%%%%%%%%%%%%%%%%%%%%%%%%%%%%%%%%%
% Optimal - deterministic
%%%%%%%%%%%%%%%%%%%%%%%%%%%%%%%%%%%%%%%%%%%%%
\subsection{Lower Bound on Optimal Cost}\label{sec:optimal}
In this section we provide a lower bound on the optimal cost, i.e., the cost of an optimal schedule.

\begin{theorem}\label{thm:lower_bound_optimal}
For any $c$-schedule $\O$ with finite cost we have
$$\cost{\O,\pi} \geq \max\set{\sum_{i=1}^n \pi_i, \frac{1}{2c} \paran{\sum_{i=1}^n \sqrt{\pi_i}}^2}.$$ 
\end{theorem}

\begin{proof}
First, by Corollary~\ref{cor:sum_lower}, $\cost{\O,\pi} \geq \sum_{i=1}^n \pi_i$. Now we show $\cost{\O,\pi} \geq \frac{1}{2c} \paran{\sum_{i=1}^n \sqrt{\pi_i}}^2$. Fix a positive integer $t >0$, and
suppose during the time interval $[0,t]$, $\O$ probes node $i$ at steps $t_1,t_2,\dots, t_{n_i}$. Let $t_0=0$ and $t_{n_i +1}=t$. So, the sequence $t_0,\dots, t_{n_i +1}$ partition the interval $[0,t]$ into $n_i +1$ intervals $I_i (j)=[t_{j}+1,t_{j+1}]$, for $0 \leq j \leq n_i$, and the length of $I_i (j)$ is $\ell_i (j)=t_{j+1}-t_j$.
Applying the Cauchy-Schwartz inequality we have:
\begin{align*}
 \sum_{j=0}^{n_i}\ell_i(j)^2 \sum_{j=0}^{n_i} 1 &\geq \paran{\sum_{j=0}^{n_i} \ell_i(j)}^2 
\\  \Longrightarrow 
 \sum_{j=0}^{n_i} \ell_i(j)^2 &\geq \frac{1}{n_i+1}\paran{\sum_{j=0}^{n_i} \ell_i(j)}^2  
  =
  \frac{t^2}{n_{i}+1}
  =
  \frac{t^2}{n_i} \paran{1-\frac{1}{n_i+1}}. 
% \\
% &\Longrightarrow \sum_{j=0}^{n_i} \ell_i(j)^2 \geq \frac{T^2}{n_i+1}=\frac{T^2}{n_i} (1-\frac{1}{n_i+1}).
\end{align*}
For $t' \in I_i(j)$, $Q_i^\O(t')$ is a Poisson random variable with parameter $\pi_i (t' - t_j)$. Therefore,
\begin{align*}
	\sum_{t=1}^t \E{Q_i^\O(t')} &= \sum_{j=0}^{n_i} \sum_{t'\in I_i(j)} \E{Q_i^\O(t')} 
	 = \pi_i \sum_{j=0}^{n_i} (1+\ldots + \ell_i(j)) \\
	 &= \pi_i \sum_{j=0}^{n_i} \frac{\ell_i(j)(\ell_i(j)+1)}{2} \geq 
 	\frac{\pi_i}{2}\sum_{j=0}^{n_i} \ell_i(j)^2 
 	\geq \frac{\pi_i}{2}\frac{t^2}{n_i} \paran{1-\frac{1}{n_i+1}}.
\end{align*}\label{eq:lim}
By summing over all nodes and averaging over $t$, we have
\begin{align}
 \sum_{i=1}^n \sum_{t'=1}^t \frac{1}{t} \E{Q_i^\O(t')} 
 &\geq
 \sum_{i=1}^n \frac{1}{t}\frac{\pi_i}{2}   \frac{t^2}{n_i} \paran{1-\frac{1}{n_i+1}} 
\nonumber \\
&= \sum_{i=1}^n \frac{\pi_i}{2}\frac{t}{n_i} \paran{1-\frac{1}{n_i+1}} 
 \geq
\frac{1}{c}\paran{\sum_{i=1}^n \frac{n_i}{t}}\paran{\sum_{i=1}^n\frac{\pi_i}{2} \frac{t}{n_i} \paran{1-\frac{1}{n_i+1}}} \nonumber \\
 &\geq \frac{1}{2c} \paran{\sum_{i=1}^n\sqrt{\pi_i}\sqrt{\paran{1-\frac{1}{n_i+1}}}}^2,
\end{align}
where in the second line we use the fact that if the schedule executed $c$ probes in each step then $\sum_{i=1}^n \frac{n_i}{t}\leq c$, and the
third line is obtained by applying the Cauchy-Schwartz inequality. 
%which is obtained by using the fact that $\sum_{i=1}^{T}\frac{n_i}{T} = 1$ and 

It remains to show that for any  schedule with finite cost, and any $i$ such that $\pi_i>0$, $\lim\limits_{t\rightarrow \infty} n_i =\infty$. For sake of contradiction assume that there is a time $s$ such that the node $i$ is never probed by $\O$ at time $t > s$. So, $\E{Q_i^\O(t)} = \pi(t-s)$ and we have
	$\cost{\O, \pi} \geq \frac{1}{t}\sum_{t'=s}^t \E{Q_i^{\O}}(t)= \frac{\pi_i}{t}\frac{(t-s)(t-s-1)}{2}$ which converges to $\infty$ as $t\rightarrow \infty$, which is a contradiction.
% $ \frac{1}{t}\sum_{t'=s}^t \E{Q_i^{\O}}(t)= \frac{\pi_i}{t}\frac{(t-s)(t-s-1)}{2}$, and thus
%$\cost{\O} =  \infty$, which is a contradiction as the round-robin schedule  has finite cost (see Section~\ref{sec:model}).
%If node $i$ was probed $n_i$ times in the interval $[0,t]$ then there are at least $t/2$ steps $\tau_1,\dots, \tau_{t/2}$ such that
%node $i$ was not probed in the intervals 
%$\left[\tau_j-\left\lfloor{t}/{(2{(n_i+1))}}\right\rfloor, \tau_j\right].$
%The expected number of undetected items at node $i$ at each $\tau_j$ is at least
%$\frac{\pi_i T}{2(n_i+1)}$, and thus, the expected number of undetected items at node $i$ throughout the interval $[0,T]$ is at least
%$\frac{1}{T}\cdot\frac{\pi_i T}{2(n_i+1)}\cdot\frac{T}{2} = \frac{\pi_i T}{4(n_i+1)}$. If 
% $n_i < \frac{\pi_i T}{8n \sum_{j=1}^n \pi_j}$, the expected number of undetected items at node $i$  is at least
%$2N\sum_{j=1}^n \pi_j$, which is greater than the cost of the round-robin schedule (see Section~\ref{sec:model}), and this schedule cannot be optimal.
Hence, for all $i$, $\lim\limits_{t\rightarrow \infty} n_i =\infty$, and using \eqref{eq:lim} we obtain
$$\cost{\O,\pi} 
\geq \lim_{t\rightarrow \infty} \frac{1}{2c} \paran{\sum_{i=1}^n\sqrt{\pi_i}\sqrt{\paran{1-\frac{1}{n_i+1}}}}^2
 = \frac{1}{2c}\paran{\sum_{i=1}^n \sqrt{\pi_i}}^2,$$
which completes the proof. \qed
\end{proof}

\subsection{Deterministic $\paran{3+(c-1)/c}$-Approximation Schedule}\label{sec:approx}

We construct a deterministic 1-schedule in which each node $i$ is probed approximately every $n_i=\frac{\sum_{j=1}^n \sqrt{\pi_j}}{\sqrt{\pi_i}}$ steps, and using that, present our $\paran{3+(c-1)/c}$-approximation schedule.
For each  $i$ let $r_i$ be a nonnegative integer such that
$2^{r_i} \geq n_i > 2^{r_i -1}$, and let $\rho=\max_i r_i$. 

\begin{lemma}\label{lem:cyclic_D}
There is a $2^{\rho}$-cyclic 1-schedule $\D$ such that node $i$ is probed exactly every $2^{r_i}$ steps.
\end{lemma}
\begin{proof}
Without loss of generality assume $\sum_{i=1}^n 2^{-r_i} =1$, otherwise we can add auxiliary nodes to complete the sum to 1, with the powers ($r_i$'s) associated with the auxiliary nodes all bounded by  $\rho$. 

We prove the lemma by induction on $\rho$. If $\rho=0$, then there is only one node, and the schedule is $1$-cyclic. Now, assume the statement holds for all $\rho' < \rho$. Since the smallest frequency is $2^{-\rho}$, and the sum of the frequencies is 1, there must be two nodes, $v$ and $u$, with same frequency $2^{-\rho}$. Join the two nodes to a new node $w$ with frequency $2^{-\rho+1}$. Repeat this process for all nodes with frequency $2^{-\rho}$. We are left with a collection of nodes all with frequencies $> 2^{-\rho}$. By the inductive hypothesis there is a $\paran{2^{\rho-1}}$-cyclic schedule $\D'$ such that each node $i$ is probed exactly each $2^{r_i}$ steps. In particular a node $w$ that replaced 
$u$ and $v$ is probed exactly each $2^{-\rho+1}$ steps. 

Now, we create an $2^{\rho}$-schedule, $\D$, whose cycle is obtained by repeating the cycle of $\D'$ two times. For each probe to $w$ that replaced a pair $u,v$, in the first cycle we probe $u$ and in the second cycle we probe $v$. Thus, $u$ and $v$ are probed exactly every $2^{\rho}$ steps, and the new schedule does not change the frequency of probing nodes with frequency larger than $2^{-\rho}$. \qed 
\end{proof}

\begin{theorem}\label{det-1}
The cost of the deterministic 1-schedule $\D$ is no more than $3$ times of the optimal cost.
\end{theorem}

\begin{proof}
By Lemma~\ref{lem:cyclic_D} each node $i$ is probed exactly every $2^{r_i}$ steps. Using $2^{r_i-1} < \frac{\sum_{j=1}^n \sqrt{\pi_j}}{\sqrt{\pi_i}}$ 
we have $2^{r_i} + 1 \leq \frac{2\cdot\sum_{j=1}^n \sqrt{\pi_j}}{\sqrt{\pi_i}}+1,$
 and therefore
% the expected number of undetected items stored at node $i$ at any given time is given by
\begin{align*}
\lim_{t\rightarrow\infty} \frac{1}{t} \sum_{t'=1}^t \E{Q_i^\D(t')} &= \lim_{t\rightarrow\infty} \frac{1}{t} \frac{t}{2^{r_i} }\sum_{t'=1}^{2^{r_i}}  \E{Q_i^\D(t') }
= 
\frac{1}{2^{r_i}}\sum_{t'=1}^{2^{r_i}}  \pi_i t' = \frac{\pi_i}{2^{r_i}} \frac{2^{r_i}(2^{r_i}+1)}{2} \\
&\leq 
\frac{\pi_i}{2} \paran{\frac{2\sum_{j=1}^n \sqrt{\pi_j}}{\sqrt{\pi_i}} + 1} 
=
\sqrt{\pi_i}\cdot \sum_{j=1}^n \sqrt{\pi_j} + \frac{\pi_i}{2}.
\end{align*}
Thus by Theorem~\ref{thm:lower_bound_optimal}, we have
\begin{align*}
 \cost{\D,\pi} &= \lim_{t\rightarrow\infty} \frac{1}{t} \sum_{i=1}^n \sum_{t'=1}^t \E{Q_i^\D(t')} 
 \leq \sum_{i=1}^n \paran{\sqrt{\pi_i}\cdot \sum_{j=1}^n\sqrt{\pi_j}} + \frac{1}{2}\sum_{i=1}^n \pi_i\\
 &= \paran{\sum_{j=1}^n \sqrt{\pi_j}}^2 + \frac{1}{2}\sum_{j=1}^n \pi_j
%  =  \frac{3}{2}\paran{\sum_{j=1}^n \sqrt{\pi_j}}^2
  \leq 3\cdot\cost{\O,\pi},
\end{align*}
where $\cost{\O,\pi}$ is the optimal cost. \qed
\end{proof}

Using the previous deterministic 1-schedule, the following corollary provides a $c$-schedule whose cost is within $\paran{3+(c-1)/c}$ factor of the optimal cost.
%and an upper bound on its cost. % as a function of the optimal cost.
%Replacing consecutive $c$ probes in the 1-schedule $\D$ with one step of a $c$-schedule, we obtain a deterministic $c$-schedule with cost that is  within a factor of 2 of the optimal for any $c$-schedule. (The only non-trivial case is when the same node is probe more than once in a sequence of $c$ probes. Details omitted for lack of space.)

\begin{corollary}\label{corr-c}
There is a deterministic $c$-schedule $\D^c$ whose cost is at most $\paran{3+(c-1)/c}$ times of the optimal cost.
%
%
%such that
%$$\cost{\D^c,\pi} \leq \frac{3}{2c}\paran{\sum_{i=1}^n \sqrt{\pi_i}}^2+\sum_{i=1}^n \frac{(c-1)\pi_i}{c} \leq \paran{3+\frac{c-1}{c}}\cost{\O,\pi},$$
\end{corollary}
\begin{proof}
Consider the execution of the deterministic 1-schedule $\D$ constructed in Theorem~\ref{det-1} on generating vector $\frac{1}{c}\pi$.
Let $\D^c$ be a deterministic $c$-schedule obtained by grouping $c$ consecutive probes of $\D$ into one step. Suppose $\O$ is an optimal $c$-schedule.
Applying equation~(\ref{eq:cost}), 
\begin{align*}
\cost{\D^c,\pi}&= \sum_{i=1}^n \pi_i \omega_i^{\D^c} = \sum_{i=1}^n \frac{\pi_i}{c} c \omega_i^{\D^c} 
\leq \sum_{i=1}^n \frac{\pi_i}{c} ( \omega_i^{\D}+c-1)  \\
&= 
\cost{\D,\pi} +\sum_{i=1}^n  \frac{(c-1)\pi_i}{c} 
\leq 3\paran{\sum_{i=1}^n\sqrt{\frac{\pi_i}{c}}}^2 +\sum_{i=1}^n \frac{(c-1)\pi_i}{c} 
\\&=\frac{3}{2c}\paran{\sum_{i=1}^n \sqrt{\pi_i}}^2+\sum_{i=1}^n  \frac{(c-1)\pi_i}{c} \leq \paran{3+\frac{c-1}{c}}\cost{\O,\pi},
\end{align*}
where the first inequality holds because  some items could be detected in less than $c$ steps in the $1$-schedule but are counted in one full step of the $c$-schedule. 
%where the first inequality holds since $\omega_i^{\D^c} \leq \left\lceil \frac{\omega_i^{\D}}{c}\right\rceil$ (this is because some items could be detected in less than $c$ steps in the $1$-schedule but are counted in one full step of the $c$-schedule), and thus,  $c \omega_i^{\D^c}\leq \omega_i^{\D}+c-1$. 
The last inequality is obtained by applying Theorem~\ref{thm:lower_bound_optimal}. \qed
\end{proof}

\subsection{On Optimal Memoryless Schedule}\label{sec:optmem}
Here, we consider memoryless schedules, and show that the memoryless 1-schedule with minimum cost can be easily computed. We call a memoryless schedule with minimum cost among memoryless schedules, an optimal memoryless schedule. We also provide an upper bound on the minimum cost of a memoryless $c$-schedule.

\begin{theorem}\label{thm:randomized_schedule}
 Let $\R = (p_1,\ldots,p_n)$ be a memoryless 1-schedule. Then $\cost{\R, \pi} \geq \paran{\sum_{i=1}^n \sqrt{\pi_i}}^2$, and the equality holds if and only if $p_i = \frac{\sqrt{\pi_i}}{\sum_{j=1}^n \sqrt{\pi_j}}$, for all $i$.%where $\mathbf{p} = (p_1,\ldots,p_N)$. We have 
% \begin{itemize}[noitemsep,nolistsep]\vspace{-3mm}
% \item[(i)]   $\cost{\R} = \sum_{i=1}^n \frac{\pi_i}{p_i}$,
% \item[(ii)]  $\cost{\R} \geq \paran{\sum_{i=1}^n \sqrt{\pi_i}}^2$, and the equality holds if and only if $p_i = \frac{\sqrt{\pi_i}}{\sum_{j=1}^n \sqrt{\pi_j}}$, for all $p_i$'s.
% \end{itemize}
\end{theorem}

\begin{proof}
Since probing each node $i$ is a geometric distribution with parameter $p_i$, the expected time until an item generated in node $i$ is discovered,
%waiting time (see Section~\ref{sec:model}) of an item generated at node $i$ 
is $\omega_i^\R = 1/p_i$. Therefore, by Lemma~\ref{lem:equivalent_cost}, we have $\cost{\R, \pi} = \sum_{i=1}^n \frac{\pi_i}{p_i}$.
We find $\ps = \argmin_{\S=p} \cost{\R, \pi}$, using the Lagrange multipliers:
\begin{align*}
 \frac{\partial}{\partial p_j}\paran{\sum_{i=1}^n \frac{\pi_i}{p_i} + \lambda \sum_{i=1}^n p_i} = 0
 \Longrightarrow p_j \propto \sqrt{\pi_j}.
\end{align*}
Therefore, $\cost{\R, \pi}$ is minimized if $p_i = \frac{\sqrt{\pi_i}}{\sum_{j=1}^n \sqrt{\pi_j}}$, and in this case the (minimized) cost will be
$$\cost{\R, \pi} = \sum_{i=1}^n \paran{\sqrt{\pi_i}\cdot\sum_{j=1}^n\sqrt{\pi_j}} = \paran{\sum_{i=1}^n \sqrt{\pi_i}}^2. \ \qed$$
\end{proof}

\begin{corollary}\label{cor:randapp}
	The cost of the optimal memoryless 1-schedule is within a factor of 2 of the cost of any optimal 1-schedule.
\end{corollary}
\begin{proof} 
	The cost of the schedule $\R$ in Theorem~\ref{thm:randomized_schedule} is $\paran{\sum_{i=1}^n \sqrt{\pi_i}}^2$, which is bounded by $2\cdot\cost{\O, \pi}$ for an optimal 1-schedule $\O$ using Theorem~\ref{thm:lower_bound_optimal}. \qed
\end{proof}

\begin{corollary}
	There is memoryless $c$-schedule, $\R^c$, whose cost is within a factor of $\paran{2+(c-1)/c}$ of any optimal $c$-schedule.
%	such that
%\begin{align*}
%cost( \R^c ,\pi) &\leq \frac{1}{c}\paran{\sum_{i=1}^n \sqrt{\pi_i}}^2+\frac{c-1}{c}\sum_{i=1}^n \pi_i.
%\end{align*}
\end{corollary}
\begin{proof}
Suppose $\R^c$ is a memoryless $c$-schedule obtained  by choosing $c$ probes in each step, each chosen according to the optimal memoryless 1-schedule, $\R$, computed in 
Theorem~\ref{thm:randomized_schedule}. Using the same argument as in the proof of Corollary~\ref{corr-c} we have
\begin{align*}
\cost{\R^c ,\pi} &\leq \frac{1}{c}\paran{\sum_{i=1}^n \sqrt{\pi_i}}^2+\frac{c-1}{c}\sum_{i=1}^n \pi_i \leq
\paran{2+\frac{c-1}{c}}\cost{\O,\pi},
\end{align*}
for an optimal $c$-schedule $\O$. \qed
\end{proof}

%since the optimal randomized schedule in Theorem~\ref{thm:randomized_schedule} is a 2-approximation schedule. 
%
%Using Theorem~\ref{thm:lower_bound_optimal} and Theorem~\ref{thm:randomized_schedule}, we have the following immediate corollary:
%\begin{corollary}\label{2_approximation}
%$\R=R(\mathbf{p})$ is a 2-approximation probing schedule when $\pi_i \propto \sqrt{\pi_i}$, for $1\leq i \leq N$. Also, $\R$ is an optimal \underline{randomized schedule}.
%\end{corollary}

% \begin{proof}
%  Using Theorem~\ref{thm:lower_bound_optimal} and Theorem~\ref{thm:randomized_schedule} 
%  we have
%  $$\cost{\R} = \paran{\sum_{i=1}^n \sqrt{\pi_i}}^2 = 2\cdot \paran{\frac{1}{2}\paran{\sum_{i=1}^n \sqrt{\pi_i}}^2} \leq 2\cdot \cost{\O},$$
%  where $\O$ is an optimal schedule.
% \end{proof}

%%%%%%%%%%%%%%%%%%%%%%%%%%%%%%%%%%%%%%%%%%%%%%%%%%%%%%
%%%%%%%%%%%%%%%%%%%%%%%%%%%%%%%%%%%%%%%%%%%%%%%%%%%%%%

\subsection{On Adaptive Algorithm for Memoryless Schedules}
Assume now that the scheduling algorithm starts with no information on the generating vector $\pi$ (or that the vector has changed). We design and analyze an adaptive 
algorithm, {\ada}, that outputs a schedule $\A$ convergent to the optimal memoryless algorithm $\R$ (see Section~\ref{sec:optmem}) by gradually learning the vector $\pi$ by observing the system. To simplify the presentation we present and analyze a 1-schedule algorithm. The results easily scale up to any integer $c>1$, where the adaptive algorithm outputs a $c$-schedule convergent to $\R^c$ (as in Section~\ref{sec:optmem}).

Each iteration of the algorithm {\ada} starts with an estimate $\tpi=(\tpi_1,\ldots,\tpi_n)$ of the unknown generating vector $\pi=(\pi_1,\dots, \pi_n)$.
Based on this estimate the algorithm chooses to probe node $i$ with probability $p_i(t) = \frac{\sqrt{\tpi_i}}{\sum_{j=1}^n\sqrt{\tpi_j}}$ (which is the optimal memoryless schedule if $t\pi$ was the correct estimate). If nodes $i$ is probed at time $t$, the estimate of $\pi_i$ is updated to
$\tpi_{i_0} \leftarrow \frac{\max(1,c_{i_0})}{t}$, where $c_{i_0}$ is the total number of new items discovered in that node since time 0.

%At any step of it execution As we saw in the previous section, when we have access to the system's parameters, i.e $\pi_i$'s, we can compute the optimal memoryless 1-schedules. In this section we study the problem in the case we do not have a prior knowledge about the generating probabilities.
%
%Here, our goal is to design a schedule that converges to the optimal memoryless 1-schedule. The idea is that as we probe the nodes, we update our estimate of $\pi_i$'s and adapt our probing probabilities. We present our algorithm {\ada} in Algorithm~\ref{alg:ada}: for a time $t$, $\ada(t)$ is a random draw from a distribution $p(t) = (p_1(t), \ldots, p_n(t))$. In fact, $\ada(t)$ is a schedule $\A(x,t)=x$ where $x$ is a random draw from $p(t)$. For notational convenience, we denote {\ada} schedule by $\A$.
%
%%%%%%%% ALGORITHM:

%\begin{wrapfigure}{L}{0.55\textwidth}
%\begin{minipage}{0.55\textwidth}
\begin{algorithm}[H]
\BlankLine
%{\bf Inputs:} Time $t$.
{\bf Outputs:} $\A(t)$, for $t=1,2,\ldots$.

\Begin{
	$(c_1,\ldots, c_n) \leftarrow (0,\ldots,0)$\;
	$(\tpi_1,\ldots,\tpi_n)\leftarrow (1,\ldots,1)$\;
	\For{$t = 1,2,\ldots $} {
		\For{$i\in\set{1,\ldots,n}$}{
			$p_i(t) \leftarrow \frac{\sqrt{\tpi_i}}{\sum_{j=1}^n\sqrt{\tpi_j}}$\;
		}
		$\A(t) \sim p(t)$\;
		\textbf{output} $\A(t)$\;	
		$c' \leftarrow$ number of new items caught at $i_0$\;
		$c_{i_0} \leftarrow c_{i_0} + c'$\;
		$\tpi_{i_0} \leftarrow \frac{\max(1,c_{i_0})}{t}$\; 
	}
}
\caption{$\ada$}\label{alg:ada}
\end{algorithm} 
%\end{minipage}
%\end{wrapfigure}

%%%%%%%% ALGORITHM:

We denote the output of {\ada} schedule by $\A$ and the optimal memoryless 1-schedule by
 $\R = \ps= (\ps_1,\ldots,\ps_n)$; see Section~\ref{sec:optmem}.
%let $T(\epsilon,\delta) = \frac{6\log(4n/\delta)}{\pi_*\epsilon^2}$, where $\pi_* = \min\set{\pi_1,\ldots,\pi_n}$ and $\epsilon, \delta \in (0,1)$. 
Our main result of this section is the following theorem.

\begin{theorem}\label{thm:convergence}
	The schedule $\A$ converges  to $\R$, and thus, $\cost{\A,\pi} = \cost{\R,\pi}$.
\end{theorem}

To prove Theorem~\ref{thm:convergence} we need the following lemmas. 

\begin{lemma}\label{lem:lower_bound_prob}
	For any time $t$ and $i\in[n]$ we have $p_i(t) \geq \frac{1}{n\sqrt{t}}$.
\end{lemma}
\begin{proof}
	It is easy to see that $p_i(t)$ will reach its lowest value at time $t$ only if for $j\neq i$ we have $\tpi_j = 1$ and $\tpi_i=\frac{1}{t-1}$ (which requires $i$ to be probed at time $t-1$). Therefore, 
%	\begin{align*}
$
		p_i(t) \geq \frac{1/\sqrt{t-1}}{1/\sqrt{t-1} + n-1} = \frac{1}{1+(n-1)\sqrt{t-1}} \geq \frac{1}{n\sqrt{t}}. 
$\qed
%	\end{align*}\qed
\end{proof}

Define $\delta(t) = 4n\exp\paran{-\frac{\pi_*t^{1/3}}{6}}$ and let $N_0$ be the smallest integer $t$ such that $\exp\paran{-\frac{\sqrt{t}}{2n}} \leq 2\exp\paran{-\frac{\pi_* t^{1/3}}{6}}$. Note that one can choose $\delta(t) = 4n\exp\paran{-\frac{\pi_*t^{1/2-\epsilon}}{6}}$ for any $\epsilon \in (0,1/2)$, and for convenience we chose $\epsilon = 1/6$.
\begin{lemma}\label{lem:all_probed}
	For any time $t \geq N_0$, with probability $\ge 1-\delta(t)/2$, all the nodes are probed during the time interval $[t/2,t)$ .
\end{lemma}
\begin{proof} %Obviously, there is a constant $N$ such that if $t\geq N$ we have $\exp\paran{-\frac{\sqrt{t}}{2n}} \leq 2\exp\paran{-\frac{\pi_* t^{1/3}}{6}}$. 
	By Lemma~\ref{lem:lower_bound_prob}, the probability of not probing $i$ during the time interval $[t/2, t)$ is at most
	$$\prod_{t'=t/2}^{t-1} (1-p_i(t')) \leq \paran{1-\frac{1}{n\sqrt{t}}}^{t/2} \leq e^{-\frac{t}{2n\sqrt{t}}} \leq 2\exp\paran{-\frac{\pi_* t^{1/3}}{6}} = \frac{\delta(t)}{2n}.$$
	A union bound over all the nodes completes the proof. \qed
\end{proof}
%For the rest of this section, we let $N_0$ be the constant that satisfies the condition in Lemma~\ref{lem:all_probed}.

\begin{lemma}\label{lem:chernoff}
	Suppose node $i$ is probed at a time $t' > t/2$. Then, 
	$$\pr{|\tpi_i(t')-\pi_i| > t^{-\frac{1}{3}}\pi_i} < \frac{\delta(t)}{2n}.$$
\end{lemma}
\begin{proof}
We estimate $\pi$ from $t'>t/2$ steps, each with $\pi_i$ expected number of new items. Applying a Chernoff bound~\cite{mu05} for the sum of $t'$ independent random variables with either Bernulli or Poisson distribution we have 
%	Note that at each time, generating an item at node $i$ happens with probability $\pi_i$ and our estimate $\tilde \pi_i$ is obtained by at least $t/2$ Bernoulli trials. Hence by Chernoff we have
\begin{align*}
	\pr{|\tpi_i(t')-\pi_i| > t^{-\frac{1}{3}}\pi_i} 
	&< 
	2\exp\paran{-\frac{t^{-\frac{2}{3}} \pi_i t'}{3}} 
	\leq
	2\exp\paran{-\frac{t^{-\frac{2}{3}} \pi_* t}{6}} 
	= \frac{\delta(t)}{2n}. \ \qed
\end{align*} 
\end{proof}

Note that by union bound, Lemma~\ref{lem:chernoff} holds, with probability at least $1-\delta(t)/2$, for all the nodes that are probed after $t/2$.

\begin{lemma}\label{lem:bounds_on_prob}
	Suppose $t \geq N_0$. With probability at least $1-\delta(t)$ we have for all $i\in [n]$,
	$$\paran{1-\frac{1}{t^{1/3}+1}}\ps_i\leq \sqrt{\frac{1-t^{-1/3}}{1+t^{-1/3}}} \ps_i \leq p_i(t) \leq \sqrt{\frac{1+t^{-1/3}}{1-t^{-1/3}}} \ps_i \leq \paran{1+\frac{1}{t^{1/3}-1}}\ps_i$$ 
	
\end{lemma}
\begin{proof}
	Applying Lemma~\ref{lem:all_probed}, Lemma~\ref{lem:chernoff} and a union bound, with probability $1-\delta(t)$  all the nodes are probed during the time $[t/2, t)$ and $|\tpi_i(t) - \pi_i| \leq t^{-1/3}\pi_i$ for all $i\in [n]$.
%	
%	By Lemma~\ref{lem:all_probed}, all the nodes are probed during the time interval $[t/2, t)$ with at least $1-\delta(t)/2$ probability, and using Lemma~\ref{lem:chernoff} and applying the union bound for all the nodes,  with probability $1-\delta(t)/2$ for all $i\in [n]$ we have $|\tpi_i(t) - \pi_i| \leq t^{-1/3}\pi_i$. Therefore, by union bound, with at least $1-\delta(t)$ probability, all the nodes are probed during the time $[t/2, t)$ and $|\tpi_i(t) - \pi_i| \leq t^{-1/3}\pi_i$ for all $i\in [n]$. Finally, 
	Since $p_i(t) = \frac{\sqrt{\tpi_i(t)}}{\sum_j \sqrt{\tpi_j(t)}}$, we obtain 
	$$p_i(t) \geq \frac{\sqrt{(1-t^{-1/3})\pi_i}}{\sum_j \sqrt{(1+t^{-1/3})\pi_j}} = \sqrt{\frac{1-t^{-1/3}}{1+t^{-1/3}}} \frac{\sqrt{\pi_i}}{\sum_j \sqrt{\pi_j}} =
		\sqrt{\frac{1-t^{-1/3}}{1+t^{-1/3}}} \ps_i \geq (1-\frac{1}{t^{1/3}+1})\ps_i
	$$
	where the last inequality uses the Taylor series of $\sqrt{1+x}$.
	 The upper bound is obtained by a similar argument. \qed
\end{proof}

\begin{corollary}
The variation distance between the distribution $p(t)=(p_1(t),\dots,p_n (t))$ used by algorithm {\ada}  at time $t\geq N_0$, and the distribution $\ps=(\ps_i ,\dots,\ps_n)$ used by the optimal memoryless algorithm satisfy
$$\parallel p(t) - \ps \parallel = \frac{1}{2}\sum_{i=1}^n |p_i(t) - \ps_i | \leq \frac{n}{t^{1/3}-1} + \delta(t) \stackrel{t \rightarrow \infty} \longrightarrow 0.$$
\end{corollary}

Finally, we present our proof for Theorem~\ref{thm:convergence}.
\paragraph{Proof of Theorem~\ref{thm:convergence}.}
Recall that we defined $\tau_i^\S(t)$ as the number of steps from the last time that node $i$ was probed until time $t$ in an execution of an schedule $\S$, and $\E{Q_i^\S} = \pi_i\E{\tau_i^\S(t)}$.

Let $F(t)$ indicate the event that the inequalities in Lemma~\ref{lem:bounds_on_prob} are held for $\forall t'\in [t/2, t)$. Therefore, $\pr{F(t)} < 1 -  \frac{\delta(t/2) \cdot t}{2}$ by applying union bound over all $t' \in [t/2,t)$, and using the fact that $\delta(t') \leq \delta(t/2)$.
Therefore,
%
%
%event that all the nodes are probed during the time interval $[t/2, t)$, and for every $i\in[n]$. Therefore, 
\begin{align*}
	|\E{Q^\A(t)} - \E{Q^\R(t)}| 
	&= \left|\sum_{i=1}^n \pi_i \E{\tau_i^{\A}(t)}-\sum_{i=1}^n \pi_i \E{\tau_i^{\R}(t)} \right| \\
	&\leq 
	\sum_{i=1}^n \pi_i \left| \E{\tau_i^\A(t)} - \E{\tau_i^\R(t)} \right| \\
	&\leq 
	\sum_{i=1}^n \pi_i \left| \E{\tau_i^\A(t)} - \E{\tau_i^\A(t)\mid F(t)} \right| \\ &+  
	\sum_{i=1}^n \pi_i \left| \E{\tau_i^\A(t)\mid F(t)} - \E{\tau_i^\R(t)} \right|,
\end{align*}
where we used the triangle inequality for both inequalities. So, it suffices to show that for every $i$, 
$$\lim\limits_{t\rightarrow\infty}\left|\E{\tau_i^\A(t)} - \E{\tau_i^\A(t)\mid F(t)}\right|
=
\lim\limits_{t\rightarrow\infty}\left|\E{\tau_i^\A(t)\mid F(t)} - \E{\tau_i^\R(t)}\right|=0.$$
Obviously, $\tau^\A_i(t) \leq t$. Now by letting $t \geq 2N_0$ we have,
\begin{align}\label{eq:upper_tau}
\E{\tau_i^\A(t)} &= \pr{F(t)}\E{\tau_i^\A(t) \mid F(t)} + \pr{\neg F(t)}\E{\tau_i^\A(t) \mid \neg F(t)} \nonumber \\
&\leq  \E{\tau_i^\A(t) \mid F(t)} + \frac{\delta(t/2)t}{2} t = 
	\E{\tau_i^\A(t) \mid F(t)} + \frac{\delta(t/2)t^2}{2}.
\end{align}
We also get
\begin{align}\label{eq:lower_tau}
\E{\tau_i^\A(t)} &\geq \paran{1-\frac{\delta(t/2)t}{2}} \E{\tau_i^\A(t) \mid F(t)}  \\
 &= \E{\tau_i^\A(t) \mid F(t)} - \frac{\delta(t/2)t}{2}\E{\tau_i^\A(t) \mid F(t)} \geq \E{\tau_i^\A(t) \mid X} - \frac{\delta(t)t^2}{2}. \nonumber
\end{align}
Note that $\lim\limits_{t\rightarrow\infty} \frac{\delta(t/2) t^2}{2} = \lim\limits_{t\rightarrow\infty} 4ne^{-\frac{\pi_*t^{1/3}}{6\sqrt[3]{2}}} t^2 = 0$, and thus by  \eqref{eq:upper_tau} and \eqref{eq:lower_tau} we have
\begin{align}\label{eq:first_term}
\lim_{t\rightarrow\infty}  \E{\tau_i^\A(t)} - \E{\tau_i^\A(t)\mid F(t)} = 0 \Rightarrow
	\lim_{t\rightarrow\infty} \left|\E{\tau_i^\A(t)} - \E{\tau_i^\A(i) \mid F(t)} \right| = 0.
\end{align}

Now, we show that $\lim\limits_{t\rightarrow\infty}\left|\E{\tau_i^\A(t)\mid F(t)} - \E{\tau_i^\R(t)}\right|=0$.
So here, we assume $F(t)$ holds. So for every $i\in[n]$, node $i$ is  probed in $[t/2,t)$, and for all $t'\in[t/2, t)$ we have
\begin{itemize}
	\item[(i)] $p_i(t') \geq \paran{1-\frac{1}{t'^{1/3}+1}}\ps_i \geq \paran{1-\frac{1}{(t/2)^{1/3} + 1}}\ps_i$. So, 
	$$\E{\tau_i^\A(t)\mid F(t)} \leq \paran{1-\frac{1}{(t/2)^{1/3} + 1}}^{-1}\frac{1}{\ps_i} = \paran{1+(t/2)^{-1/3}}\frac{1}{\ps_i}.$$
	
	\item[(ii)] $p_i(t') \leq \paran{1+\frac{1}{t'^{1/3}-1}}\ps_i \leq \paran{1+\frac{1}{(t/2)^{1/3} - 1}}\ps_i$. Hence,
	$$\E{\tau_i^\A(t)\mid F(t)} \geq \paran{1+\frac{1}{(t/2)^{1/3} - 1}}^{-1}\frac{1}{\ps_i} = \paran{1-(t/2)^{-1/3}}\frac{1}{\ps_i}.$$
\end{itemize}
Obviously $\E{\tau_i^\R(t)} = \frac{1}{\ps_i}$, since probing  node $i$ by $\R$ can be viewed as a geometric distribution with parameter $\ps_i$, and since $\delta(t)\rightarrow 0$ as $t\rightarrow \infty$ we have

\begin{align*}
\E{\tau_i^\R(t)} &= \frac{1}{\ps_i} = 
\lim_{t\rightarrow \infty} \paran{1-(t/2)^{-1/3}}\frac{1}{\ps_i} \leq \lim_{t\rightarrow \infty} \E{\tau_i^\A(t) \mid F(t)} \\
&\leq 
\lim_{t\rightarrow\infty} \paran{1+(t/2)^{-1/3}}\frac{1}{\ps_i} = \frac{1}{\ps_i} = \E{\tau_i^\R(t)}.
\end{align*}
Therefore, 
\begin{align}\label{eq:second_term}
\lim_{t\rightarrow\infty} |\E{\tau_i^\A(t)\mid F(t)}- \E{\tau_i^\R(t)}| = 0.
\end{align}
Thus, by \eqref{eq:first_term} and \eqref{eq:second_term} we have $\lim\limits_{t\rightarrow\infty} |\E{Q^\A(t)} - \E{Q^\R(t)}| = 0$, and $\A$ converges to $\S$, and since $\lim\limits_{t\rightarrow \infty} \E{Q^\S(t)}  = \paran{\sum_{i=1}^n \sqrt{\pi_i}}^2$, it implies that $\lim\limits_{t\rightarrow \infty} \E{Q^\A(t)} = \paran{\sum_{i=1}^n \sqrt{\pi_i}}^2 = \cost{\A, \pi}$ (by Cesaro Mean~\cite{hardy1991divergent}). \qed

Note that one can obtain an adaptive schedule $\A^c$ by choosing $c$ probes in each step, at each round of {\ada}, and using similar argument as in Section~\ref{sec:optmem} (and similar to Corollary~\ref{corr-c}), it is easy to see that $\A^c$ converges to $\R^c$.

Finally, if $\pi$ changes, the {\ada} algorithm converges to the new optimal memoryless algorithm, as the change in the rate of generating new items is observed by {\ada}.

\bibliographystyle{llncs2e/splncs03}
{\footnotesize \bibliography{iso}}

\end{document}